\documentclass[AMA,Times1COL]{WileyNJDv5} 
\usepackage{subcaption}
\usepackage{graphicx}
\usepackage{epstopdf}
\articletype{Article Type}%

\received{Date Month Year}
\revised{Date Month Year}
\accepted{Date Month Year}
\journal{Journal}
\volume{00}
\copyyear{2023}
\startpage{1}

\begin{document}

\title{Geometric Tracking Control of a Multi-rotor UAV for Partially Known Trajectories}

\author[1]{Yogesh Kumar}

\author[1]{S.B. Roy}

\author[2]{P.B. Sujit}

\authormark{Yogesh \textsc{et al.}}
\titlemark{Geometric Tracking Control of a Multi-rotor UAV for Partially Known Trajectories}

\address[1]{\orgdiv{ Department of Electronics and Communication Engineering}, \orgname{Indraprastha Institute of Information Technology Delhi}, \orgaddress{\state{Delhi}, \country{India}}}

\address[2]{\orgdiv{Department of Electrical Engineering and Computer Science}, \orgname{Indian Institute of Science Education and Research Bhopal}, \orgaddress{\state{Madhya Pradesh}, \country{India}}}

\corres{Corresponding author Yogesh Kumar, \email{yogeshk@iiitd.ac.in}}


\abstract[Abstract]{This paper presents a trajectory-tracking controller for multi-rotor unmanned aerial vehicles (UAVs) in scenarios where only the desired position and heading are known without the higher-order derivatives. The proposed solution modifies the state-of-the-art geometric controller, effectively addressing challenges related to the non-existence of the desired attitude and ensuring positive total thrust input for all time. We tackle the additional challenge of the non-availability of the higher derivatives of the trajectory by introducing novel nonlinear filter structures. We formalize theoretically the effect of these filter structures on the system error dynamics. Subsequently, through a rigorous theoretical analysis, we demonstrate that the proposed controller leads to uniformly ultimately bounded system error dynamics.}


\jnlcitation{\cname{%
\author{Kumar Y},
\author{Roy SB}, and
\author{Sujit PB}}.
\ctitle{Geometric Tracking Control of a Multi-rotor UAV for Partially Known Trajectories.} \cjournal{\it Int J Robust Nonlinear Control.} \cvol{2023;xx(xx):x--xx}.}

\maketitle


\renewcommand\thefootnote{\fnsymbol{footnote}}
\setcounter{footnote}{1}

\section{Introduction} \label{sec:Intro}
Multi-rotor unmanned aerial vehicles (UAVs) are becoming a commonplace operational necessity due to their dynamic capabilities, such as vertical take-off, landing, and hovering. Over the years, many developments in multi-rotor controller synthesis have occurred, especially for quadrotors considering various factors, including implementation complexity, required performance, and available computational resources \cite{UAV_Survey}. The commonly used multi-rotor UAVs are underactuated with highly non-linear and coupled dynamics, which demands special consideration when designing globally stable trajectory tracking controller \cite{Multirotor_Basic}.

While hardware and software advances are shrinking the gaps between theory and practice, the ever-increasing application domain introduces new challenges. One such challenge is to solve a tracking control problem where the desired trajectory is not fully known, especially for cases that require non-cooperative target tracking or coordination with other vehicles \cite{PBVS_Vijay,Tenghuchang_Vis}. Due to hardware or security considerations, obtaining complete information about the motion of the other vehicle other than the pose measurement is often challenging \cite{PBVS_Yogesh}. The non-cooperative tracking problem is usually tackled using the augmentation of vision sensors, such as - position-based visual servoing (PBVS) \cite{Vision_motivation,PBVS_Vijay,PBVS_Switching} or image-based visual servoing (IBVS) \cite{Teng2, RITA_OpticalFlow, VTOL_Automatica_2019, VTOL_Automatica_CDC}.

Compensating for unknown target dynamics 
is a challenging issue that is often tackled by either introducing a robust conservative component in the controller or deriving feedback compensation for the target velocity based on ad-hoc techniques such as optical flow \cite{RITA_OpticalFlow, VTOL_Automatica_CDC,Mahony_optical} and Kalman filter \cite{PBVS_Vijay,sanchez2014approach_kalman, Teng2}, etc. while considering restrictive assumptions on the target motion. In our earlier work \cite{IBVS_Yogesh}, we proposed an IBVS method for quadrotors to overcome these challenges at the kinematic level. However, the internal attitude and translational controllers further require higher derivatives of the target dynamics. The knowledge of higher derivatives of target dynamics is thus necessary, and solutions such as high controller gain \cite{VTOL_Automatica_2019} or command filter-based \cite{Command_Filter} techniques \cite{Robust_IBVS_Landing} are proposed in the literature. Recently, the authors in \cite{Robust_Khalil} introduced a high gain controller based on the Euler angle formulation to address the requirement of higher derivatives of the target dynamics. However, Euler angle-based methods show singularities \cite{Geometric_2010}, which restricts the UAV performance for nontrivial trajectories. 

This paper proposes a solution to the general problem of designing a trajectory-tracking controller for the multi-rotor UAV when only the desired position and heading are available. We will design our controller by suitably modifying the state-of-the-art geometric controller proposed in \cite{Geometric_2010,Exponential_Geometric}. Developed on the special Euclidean group $SE(3)$, the geometric controller enables us to design an almost global trajectory tracking controller for multi-rotor UAVs. However, two major limitations of this controller are (i) the thrust is not saturated, hence no guarantee of the existence of the desired angular orientation, and (ii) the designed thrust can be negative when the angle between the third body axis of a multi-rotor UAV and desired thrust direction is greater than $90^\circ$ \cite{Disturbance_Kar}. Our proposed solution addresses these limitations in the current context of designing the controller for partially known trajectories. 

The main contributions of the paper are as follows 
\begin{itemize}
    \item  We consider the problem of multi-rotor tracking a partially known trajectory and develop a trajectory-tracking controller by suitably modifying the state-of-the-art geometric controller \cite{Geometric_2010}.

    \item Unlike \cite{NotSO(3)}, an auxiliary filter dynamics in $SO(3)$ space is proposed for the UAV attitude system, 
    which tackle the requirement of higher-order derivatives of the desired position in desired attitude derivatives. A similar approach using command filter is presented in \cite{Quaternion_ACC}; however, considering a quaternion-based formulation. To the best of our knowledge, this is the first time the consequences of such filter dynamics for attitude representations in $SO(3)$ are quantified on the overall system through a Lyapunov-based analysis.
    
    \item To ensure the non-zero desired thrust requirements for the existence of the desired attitude \cite{VTOL_Automatica_2019, Ranjan_ACC}, we propose a projection-based filter dynamics to yield a bounded second-order derivative of the desired position. Further, we discuss the effect of large attitude errors on the thrust control inputs proposed in \cite{Geometric_2010, Disturbance_Kar} in detail and introduce a novel control strategy.  
    
    \item The effect of the proposed control scheme, along with the auxiliary filter dynamics, is theoretically analyzed. The overall system consisting of translational, angular, and auxiliary dynamics is shown to be uniformly ultimately bounded through rigorous mathematical arguments. The careful design and extensive theoretical analysis enable us with an ultimate bound that can be made arbitrarily small by choosing high enough auxiliary dynamics gains without burdening the controller gains. 
\end{itemize}

The rest of the paper is organized as follows. Section \ref{Sec:Formulation} describes the problem formulation with the control objective. The controller design is presented in Section \ref{Sec:Control} followed by the theoretical analysis in Section \ref{Sec:Analysis}. Finally, section \ref{Sec:conclusion} draws the conclusions.
 
\section{Problem Formulation and Preliminaries} \label{Sec:Formulation}
\subsection{Notations}
The set of real numbers is represented as $\mathbb{R}$, while the set of positive real numbers as $\mathbb{R}_{+}$. The identity matrix of size $3$ is represented as $I$. Operators $\lambda_{max}(.)$ and $\lambda_{min}(.)$ represents the maximum and minimum eigenvalue of input matrix, respectively. The special orthogonal group $\mathrm{SO(3)}$ which contains the orthogonal matrices is defined as 
 \begin{align} \label{eq:orthogonal}
 \mathrm{SO(3)}  = \{R \in \mathbb{R}^{3 \times 3} | R^TR = RR^T =  I, det(R) = 1\}.
 \end{align}
 
 The group of skew-symmetric matrices $\mathfrak{so}(3)$ is defined as 
 \begin{align} \label{eq:skew}
  \mathfrak{so}(3)  = \{S \in \mathbb{R}^{3 \times 3} | S = -S^T\}.
 \end{align}
 Map $(\hat{.}):\mathbb{R}^3 \rightarrow \mathfrak{so}(3)$ is the skew-symmetric map whereas it's inverse is denoted by $(.)^\vee : \mathfrak{so}(3) \rightarrow \mathbb{R}^3$ and $e_3 = [0,0,1]^T$. For a vector $v$, $\|v\|$ denotes the 2-norm whereas for a matrix $M$, $\|M\|$ denotes the 2-norm until unless specified. $\ln()$ represents the natural logarithmic operator. For a vector $v = [v_1,v_2,v_3]^T$,
\begin{align}
    \tanh(v) = &[\tanh(v_1),\tanh(v_2),\tanh(v_3)]^T \in \mathbb{R}^3, \label{eq:tanh} \\
    Cosh(v) = &daig\{\cosh(v_1),\cosh(v_2),\cosh(v_3)\} \in \mathbb{R}^{3 \times 3}, \label{eq:Cosh} \\
    Tanh(v) = &daig\{\tanh(v_1),\tanh(v_2),\tanh(v_3)\} \in \mathbb{R}^{3 \times 3}. \label{eq:Tanh} \\
    Sech(v) = &Cosh^{-1}(v). \label{eq:Sech}
\end{align}

\subsection{Multi-rotor Dynamics}  \label{sec:Prelim}
 Consider a multi-rotor UAV illustrated in Figure \ref{fig:1}. Let \{I\} and \{B\} be the inertial and the vehicle body frames, with origins $O_I$ and $O_B$, respectively. The configuration of UAV is defined by the position $x \in \mathbb{R}^3$ of the center of mass and attitude $R = [{x}_B, {y}_B, {z}_B] \in \mathrm{SO(3)}$ in the inertial frame. 
 
We model the multi-rotor dynamics as
\begin{align} 
    \dot{x} = &v, \label{eq:dyn1} \\ 
    m\dot{v} = &mge_3 - fRe_3, \label{eq:dyn2} \\ 
    \dot{R} = &R\hat{\Omega}, \label{eq:dyn3} \\ 
    J\dot{\Omega} =  & -\Omega \times J\Omega + M, \label{eq:dyn4}
\end{align}
where inertial frame position $x(t) \in \mathbb{R}^3$, inertial frame velocity $v(t) \in \mathbb{R}^3$, the rotation matrix from \{B\}  to \{I\}, $R(t) \in SO(3)$ and the body-fixed angular velocity $\Omega(t) \in \mathbb{R}^3$ represents the multi-rotor states. The total thrust $f \in \mathbb{R}$ and moment $M=[M_1, M_2, M_3]^T\in \mathbb{R}^3$ expressed in body frame, are considered as the control inputs to the multi-rotor model. For a given pair of total thrust and moment, the rotor speeds can be calculated using the following relation 
\begin{align} 
    \begin{bmatrix}\label{eq:dyn5}
    f \\
    M
    \end{bmatrix} = \Gamma \omega_T, \ \omega_T = [\omega_1^2, ..., \omega_n^2]^T, 
\end{align}
where n is the number of rotors and $\Gamma \in \mathbb{R}^{4 \times n}$ is a constant mixing matrix constructed based on the rotor configuration parameters such as placement of rotors, thrust, and moment coefficient of rotors etc \cite{Multirotor_Basic}.  
 
 \begin{figure}[h!]
    \centering
    \includegraphics[width=9cm,height=7cm]{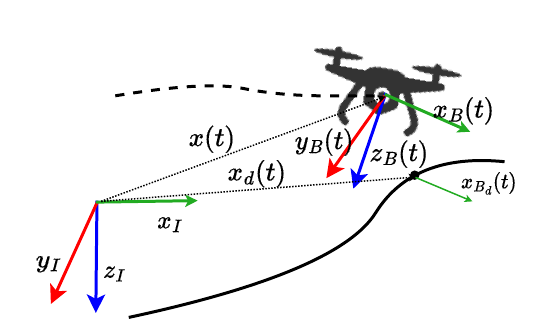} 
    \caption{Coordinate frames illustration. The $x$, $y$, $z$ axes of
a coordinate frame \{A\} are denoted by ${x}_A$, ${y}_A$ and ${z}_A$, respectively.}
    \label{fig:1}
\end{figure}
 
\subsection{Control Objective}
The objective is to design a nonlinear trajectory tracking controller for a multi-rotor UAV to track a position reference $x_d(t)$ and a heading direction $x_{B_d}(t)$ or a heading angle $\psi_d(t)$, where $x_{B_d}(t) = [cos\psi_d(t),sin\psi_d(t),0]^T$. Moreover, we consider that the higher derivatives of the position reference and heading direction ( i.e. $\frac{d^nx_d(t)}{dt^n} \ \text{and} \ \frac{d^nx_{B_d}(t)}{dt^n}, \ n \geq 1$) are not available, which are essential for an accurate and efficient trajectory tracking controller.  

\begin{assumption}\label{assump1}
We assume that the reference trajectory has bounded derivatives as  
\begin{align} \label{eq:dyn6}
    \sup_{t \geq 0}||\frac{d^n{x}_d(t)}{dt^n}|| = {h}_n, \quad n = 1,2,3,4, \quad
    \sup_{t \geq 0}||\frac{dx_{B_d}(t)}{dt}|| = {h}_5, \ \ \ \sup_{t \geq 0}||\frac{d^2x_{B_d}(t)}{dt^2}|| = {h}_6 .
\end{align}
where $h_1, h_2, h_3, h_4, h_5, h_6 \in \mathbb{R}_{+}$ are finite constants.
\end{assumption}

\section{Control Design} \label{Sec:Control}
We follow a multistage procedure illustrated in Figure \ref{fig:2} to design the trajectory tracking controller. We first design the position tracking controller to obtain the control input $f(t)$. Subsequently, a desired rotation matrix $R_c(t) \in \mathrm{SO(3)}$ constructed exploiting unit vector direction of $f$ and desired heading direction $x_{B_d}(t)$. We then pass $R_c(t)$ through an auxiliary filter dynamics to obtain the filtered desired attitude reference, which the proposed attitude tracking controller would track.   

 \begin{figure}[h!]
    \centering
    \includegraphics[width=10cm,height=6cm]{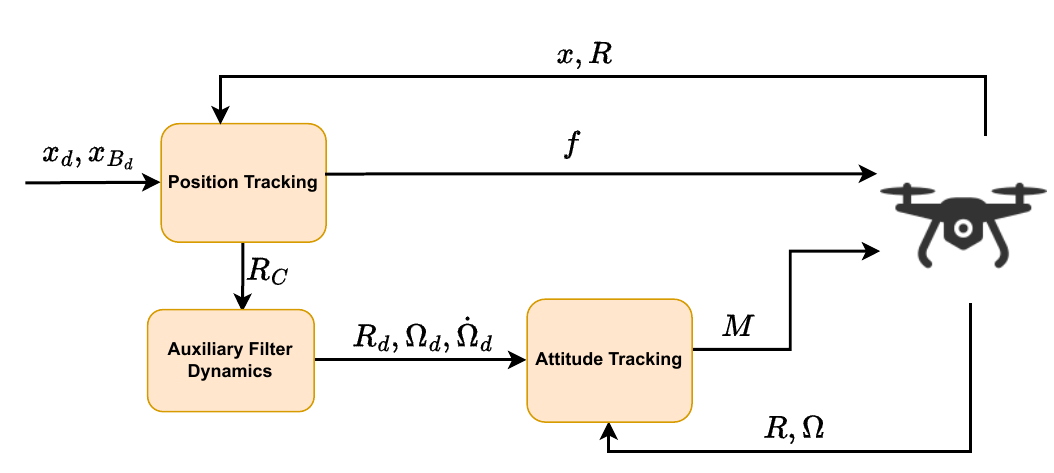}
    \caption{Control architecture}
    \label{fig:2}
\end{figure}

\subsection{Position Tracking Controller}
To design a smooth, saturated position tracking controller, we consider an error system \cite{Ranjan_ACC} with position tracking error $e_x(t)$ and filter tracking error $e_\alpha(t)$ as 
\begin{align}
    &e_x = x - x_d = [e_{x_1},e_{x_2},e_{x_3}]^T, \label{eq:pos1}\\
    &e_{\alpha} = \dot{e}_x + \alpha_{x} \tanh(e_x) + \tanh(e_f), \label{eq:pos2}
\end{align}
where $e_f(t) = [e_{f_x}(t),e_{f_y}(t),e_{f_z}(t)]^T$ is an auxiliary filter variable with the dynamics  
\begin{align}
    &\dot{e}_{f} = Cosh^2(e_f)\left(-k_{\alpha}e_{\alpha} + \tanh(e_x) - \alpha_{f}\tanh(e_f)\right), \label{eq:pos3}
\end{align}
and $\{k_{\alpha},\alpha_f, \alpha_x \in \mathbb{R}_{+}\}$ are positive constant control design parameters. 

The thrust control input $f(t)$ is designed as 
\begin{align}\label{eq:pos4}
f = \|f_d\|\sqrt{\frac{1 + e_3^TR_c^TRe_3}{2}}
\end{align}
where $f_d(t)$ is a virtual control input defined as 
\begin{align}
    &f_d = mge_3 - m{g}_{1} + 2m\tanh(e_x) - mk_{\alpha}\tanh(e_f), \label{eq:pos5}
\end{align}
and $g_1(t)$ is an estimate of $\ddot{x}_d(t)$ which is not directly available and $R_c$ is the desired attitude matrix carefully constructed in the subsequent section. Note that the thrust expression requires term $\dot{e}_x$, i.e., velocity feedback, in the $e_f$. Since we don't have the explicit knowledge of $\dot{x}_d$, one can circumvent the requirement of unknown term $\dot{e}_x$ in thrust input calculation following the output-feedback formulation in \cite{Output_Dixon}. This also makes the position dynamics loop output feedback.

\subsection{Estimate of $\ddot{x}_d(t)$}
To obtain the smooth and bounded estimate $g_1(t)$ of $\ddot{x}_d(t)$, we start by considering the following filter structure with filter input $x_d(t)$ and output $x_{f_d}(t)$
\begin{align}
    &\gamma_1\dot{x}_{f_d} + x_{f_d} = x_d, \ \ x_{f_d}(0) = 0. \label{eq:pos6}
\end{align}
Consider another filter with filter design parameter $\gamma_1 \in \mathbb{R}_{+}$ to filter the unknown variable $\dot{x}_d(t)$ as  
 \begin{align}
\gamma_1 \dot{g}_{x_d} + {g}_{x_d} = \dot{x}_d, \ \  {g}_{x_d}(0) = 0. \label{eq:pos7}
\end{align}
Integrating \eqref{eq:pos7} for ${g}_{x_d}(t)$ yield 
 \begin{align}
 {g}_{x_d}(t) = \frac{1}{\gamma_1}\left(x_d(t)-x_{f_d}(t)-exp(-\frac{1}{\gamma_1}t)x_d(0)\right). \label{eq:pos8}
\end{align}
Since ${g}_{x_d}(t)$ represents the filtered output of unknown variable $\dot{x}_d(t)$, we further pass it through a filter with design parameter $\gamma_2 \in \mathbb{R}_{+}$ as 
\begin{align}
    &\gamma_2\dot{g}_{f_d} + g_{f_d} = g_{x_d}, \ \ g_{f_d}(0) = 0, \label{eq:pos9}
\end{align}
to obtain the estimate of $\ddot{x}_d(t)$. Further following the filter structure similar to \eqref{eq:pos7} for the unknown variable $\dot{g}_{x_d}$ and integrating it for its filter variable ${g}_{v_d}(t)$, we get
\begin{align}\label{eq:pos10}
{g}_{v_d}(t) = \frac{1}{\gamma_2}\left(g_{x_d}(t)-g_{f_d}(t)-exp(-\frac{1}{\gamma_2}t)g_{x_d}(0)\right).    
\end{align}

For clarity, the filter structures discussed above are illustrated in Figure \ref{fig:3}. Filter block 1 includes structures given by equations \eqref{eq:pos6} and \eqref{eq:pos8}, with input $x_d(t)$ and output signal ${g}_{x_d}(t)$. Similarly, filter block 2 takes ${g}_{x_d}(t)$ as input and provides ${g}_{v_d}(t)$ using equations \eqref{eq:pos9} and \eqref{eq:pos10}.
\begin{figure}[h!]
    \centering
    \includegraphics[width=9cm,height=7cm]{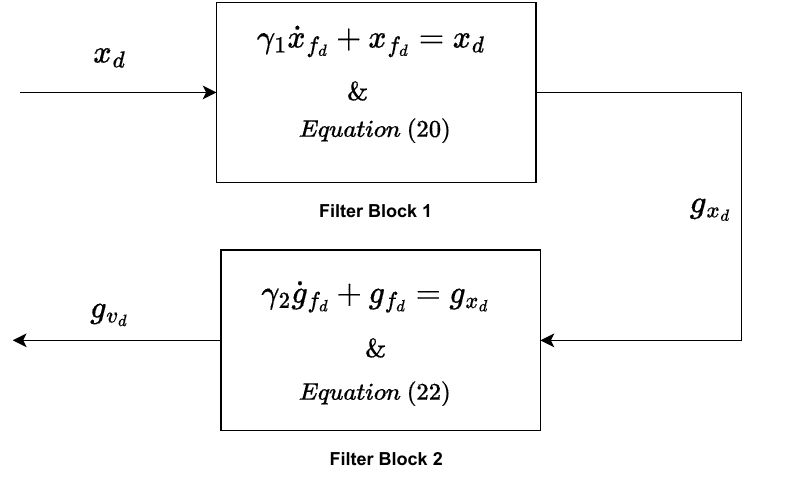} 
    \caption{Filters illustration}
    \label{fig:3}
\end{figure}

Following the rigorous analysis presented in the subsequent sections, we can use ${g}_{v_d}(t)$ as an estimate of $\ddot{x}_d(t)$, however, this estimate can suffer from peaking phenomenon due to the low value of $\gamma_2$ which might jeopardize the saturated thrust requirements to design a smooth attitude tracking controller. To get a saturated thrust input, a differential projection filter for a bounded estimate $g_1(t)$ of $\ddot{x}_d(t)$ with design parameter $\gamma_{21} \in \mathbb{R}_{+}$ and $g_1(0) = 0$ is proposed as 
\begin{equation}\label{eq:pos11}
    \resizebox{1\width}{!}{$%
    \begin{aligned}[b]
    & \dot{g}_1 = Proj(g_1,\phi) = 
     \begin{cases}
    \phi - \frac{\nabla{f(g_1)}\left(\nabla{f(g_1)}\right)^T}{||\nabla{f(g_1)}||^2}\phi f(g_1),&\text{if } f(g_1)>0 \ \& \  \phi^T\nabla{f(g_1)}>0\\
    \phi,             & \text{otherwise}
\end{cases} ,
    \end{aligned}$%
    }%
\end{equation}
where 
\begin{equation}\label{eq:pos12}
\resizebox{1\width}{!}{$%
\phi = \frac{1}{\gamma_{21}}({g}_{v_d} - g_1),
f(g_1) = \frac{||g_1||^2 - h_2^2}{\epsilon_0h_2^2}, 
\nabla{f(g_1)} = \frac{2g_1}{\epsilon_0h_2^2},\nonumber $%
}%
\end{equation}
where $h_2$ is the known upper bound of the $\ddot{x}_d$ and $\epsilon_0>0$ is the tolerance.

\subsection{Auxiliary Desired Attitude Construction}
 The desired attitude $R_c(t)$ is constructed exploiting the virtual control input $f_d(t)$ given in \eqref{eq:pos5} and the desired heading direction $x_{B_d}(t)$ as 
\begin{align}\label{eq:Rc}
&R_c = [x_{B_c},y_{B_c},z_{B_c}],  \\
&z_{B_c} = \frac{f_d}{\|f_d\|} \quad y_{B_c} = -\frac{x_{B_d} \times z_{B_c}}{\|x_{B_d} \times z_{B_c}\|} \quad x_{B_c} = y_{B_c} \times z_{B_c}.
\end{align}
 Once we construct the desired attitude matrix a well defined\footnote{We assume that $x_{B_d}(t) \nparallel z_{B_c}(t)$. We note that one can circumvent this singularity in differential
flatness transformation using the Hopf Fibration on SO(3) proposed in [25]. However this is out of scope of the current problem formulation.} $R_c(t)$ as given in \eqref{eq:Rc}, the desired angular velocity $\Omega_c(t)$ and acceleration $\dot{\Omega}_c(t)$ corresponding to this can be calculated as 
\begin{align}
\hat{\Omega}_c = &R_c^T\dot{R}_c, \label{eq:att1} \\
\dot{\hat{\Omega}}_c = &R_c^T\ddot{R}_c - \hat{\Omega}_c^2, \label{eq:att2}
\end{align}

Note that the above calculations are essential to design an accurate attitude tracking control law for the dynamics given in \eqref{eq:dyn3} and \eqref{eq:dyn4} \cite{Geometric_2010,Exponential_Geometric}. However, one requires information about the higher derivatives of the virtual input $f_d(t)$ and heading direction $x_{B_d}(t)$ to calculate $\dot{R}_c, \ddot{R}_c$. Providing the information mentioned earlier is not feasible as it involves higher derivatives of the desired trajectory and even the multi-rotor UAV's states such as rate of acceleration (jerk). This issue is tackled by systematically defining an appropriate auxiliary filter dynamics for the available signal $R_c(t) \in SO(3)$ to obtain a filtered signal $R_d(t) \in SO(3)$ strictly belonging to the special orthogonal group and an angular velocity $\Omega_d(t)$ associated to it as 
\begin{align}
\dot{R}_d = &R_d\hat{\Omega}_d, \label{eq:att3} \\
\dot{\Omega}_d = &-\frac{1}{\gamma_4}e_{\Omega_{dc}}+ e_{R}(R_d,R_c) + \frac{1}{\gamma_3}E(R_d,R_c)\Omega_d, \label{eq:att4}
\end{align}
where $e_{R}(R_d,R_c)$ is the attitude error between $R_c$ and $R_d$ as described in Appendix A. $\gamma_3, \gamma_4 \in \mathbb{R}_{+}$ are the constant design parameters. The term $E(R_d,R_c)$ is also described in Appendix \ref{secA1}. The error $e_{\Omega_{dc}} = \Omega_d + \frac{1}{\gamma_3}e_{R_{dc}}$ is carefully designed based on the stability analysis in the subsequent section. 

\subsection{Effect of large attitude errors on $f$}
Authors in \cite{Geometric_2010} discussed the necessity of designing the thrust control input $f$ to reduce its magnitude whenever there are large attitude errors. They do so by multiplying the magnitude of the virtual control input $f_d$ with the dot product of $z_{B_c} = R_ce_3$ and $z_{B} = Re_3$, i.e., $f = \|f_d\|e_3^TR_c^TRe_3$. The term $e_3^TR_c^TRe_3$ represents the cosine of the angle between $z_{B_c}$ and $z_{B}$, which can make the control input $f$ negative whenever the angle is greater than $90^\circ$, which is not feasible for the quadrotor to attain. To address this, authors in \cite{Disturbance_Kar} suggested using simply the magnitude of the virtual control input, i.e., $f = \|f_d\|$ as the thrust control input. However, this may cause instability in the overall system dynamics by generating unnecessarily high thrust for large attitude errors. The careful analysis and design of the thrust control input in equation \eqref{eq:pos4} attenuates both of these problems by multiplying the scaling quantity $\sqrt{\frac{1 + e_3^TR_c^TRe_3}{2}}$ to $\|f_d\|$ which is nothing but cosine of half of the angle between $z_{B_c}$ and $z_{B}$. An illustrative comparison of different thrust inputs with the respective virtual control inputs in the case of \cite{Geometric_2010}, \cite{Disturbance_Kar} and the proposed method w.r.t to the angle between $z_{B_c}$ and $z_{B}$ is provided in Figure \ref{Fig:Thrust}.

\begin{figure}
    \centering
    \includegraphics[width=9cm,height=6cm]{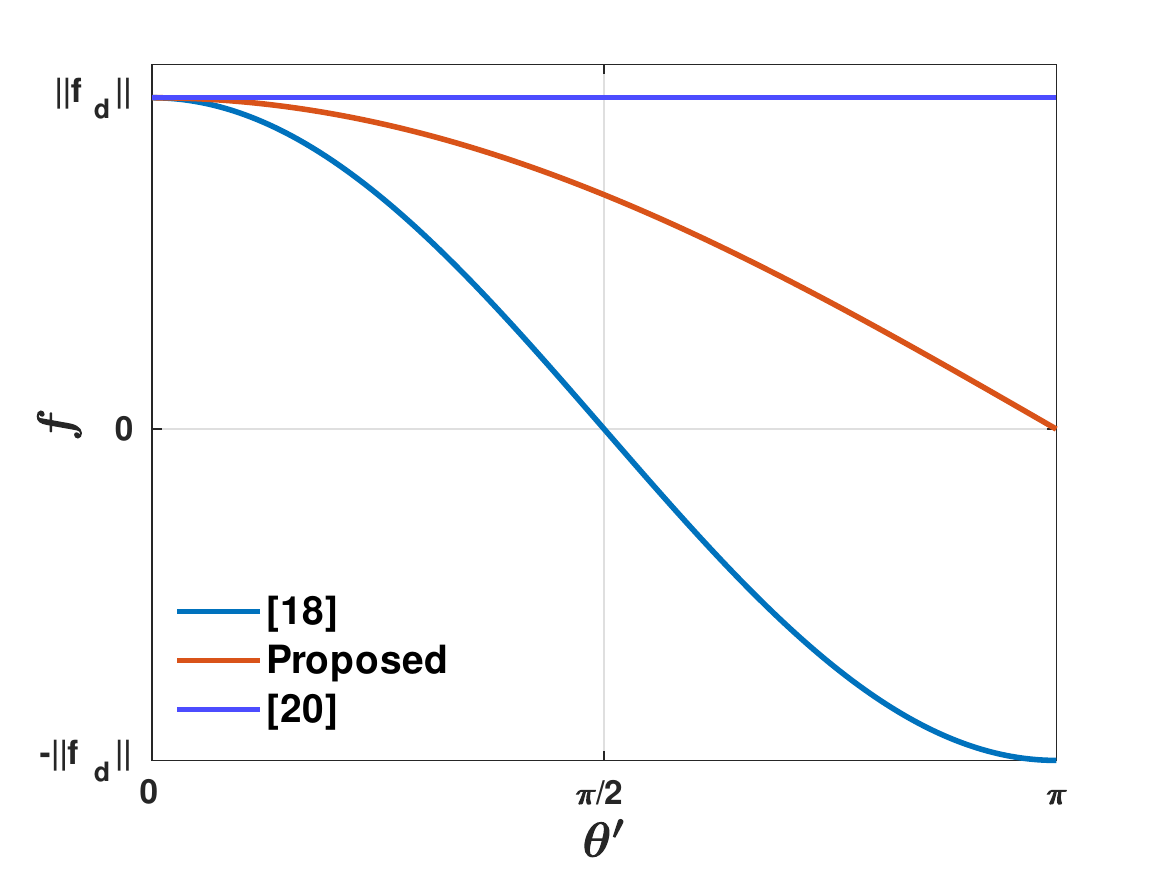}
    \caption{Thrust control input comparison. $\theta^{\prime}$ is the angle between $z_{B_c}$ and $z_{B}$.}
    \label{Fig:Thrust}
\end{figure}

\begin{remark} \label{remark1}
Note that the proposed virtual control input \eqref{eq:pos5} is always positive for a tunable $k_{\alpha}$, which is a necessary condition for the existence of the desired attitude $R_c(t)$. Moreover, the requirement of proportionally reduced thrust magnitude for large attitude error is inherently captured in the designed thrust input \eqref{eq:pos4} in the whole configuration space (especially angles greater than $90^\circ$). 
\end{remark}

\subsection{Attitude Tracking Controller}
We now design the attitude control input $M$ \cite{Exponential_Geometric} to track the auxiliary reference signals $R_d(t),\Omega_d(t),\dot{\Omega}_d(t)$ as follows   
\begin{align}\label{eq:att8}
    M = &-k_Re_R(R,R_d) - k_{\Omega}e_{\Omega}(R,\Omega,R_d,\Omega_d) + \Omega \times J\Omega - J(\hat{\Omega}R^TR_d\Omega_d - R^TR_d\dot{\Omega}_d), 
\end{align}
where $\{k_{R}, k_{\Omega} \in \mathbb{R}_{+}\}$ are constant control design parameters. The attitude and angular velocity errors $e_R(R,R_d)$ and $e_{\Omega}(R,\Omega,R_d,\Omega_d)$ are defined in detail in Appendix \ref{secA1}.

\section{Theortical Analysis} \label{Sec:Analysis}
To comment on the stability and error convergence of the complete system, we first derive the comprehensive closed-loop error dynamics for individual subsystems.

\subsection{Position Tracking Error System}
Using the relations given in \eqref{eq:pos1} and \eqref{eq:pos2} the time derivative of $e_x$ can be given as 
\begin{align} \label{eq:poser1}
    \dot{e}_x = e_{\alpha} - \alpha_{x} \tanh(e_x) - \tanh(e_f).
\end{align}

The open-loop dynamics of the filter tracking error multiplied by mass  $me_{\alpha}(t)$ substituting the system dynamics \eqref{eq:dyn2} and auxiliary filter dynamics \eqref{eq:pos3} can be written as  
\begin{align}\label{eq:poser2}
    m\dot{e}_{\alpha} = &mge_3 - fRe_3 - m\ddot{x}_d + \chi  -mk_{\alpha}e_{\alpha} + m\tanh(e_x) 
    \end{align}
where $ \chi = m\alpha_{{x}}Sech^2(e_x)\dot{e}_x - m\alpha_{f}\tanh(e_f)$. To better understand the attitude dynamics coupling, the virtual control input $f_d(t)$ is added and subtracted to the right-hand side of \eqref{eq:poser2} as 
\begin{align}\label{eq:poser3}
    m\dot{e}_{\alpha} = &mge_3 + \Delta{f} + \chi - m\ddot{x}_d -mk_{\alpha}e_{\alpha} + m\tanh(e_x) - f_d,
\end{align}
where $\Delta{f} = f_d - fRe_3$.    

Using the expression for virtual control input from \eqref{eq:pos5}, the closed-loop filter error dynamics can be written as 
\begin{align}\label{eq:poser4}
m\dot{e}_{\alpha} = & \Delta{f} + \chi - m(\ddot{x}_d - g_1) -mk_{\alpha}e_{\alpha} - m\tanh(e_x) + mk_{\alpha}\tanh(e_f).
\end{align}

To analyze the dynamical behavior of the proposed filter structures for desired position reference derivatives $\frac{d^nx_d(t)}{dt^n}$, we propose the following error system  $e_d(t)$ as 
\begin{align}\label{eq:poser5}
e_d = \begin{bmatrix}
e_{g_x} \\
e_{\ddot{x}_d}\\
e_{g_v} \\
e_{g_1} \\
\end{bmatrix} = \begin{bmatrix}
\dot{x}_d - {g}_{x_d} \\
\ddot{x}_{d} - \dot{g}_{x_d} \\
\dot{g}_{x_d} - {g}_{v_d} \\
\ddot{x}_{d} - g_1 \\
\end{bmatrix}.
\end{align}

Using the expressions from \eqref{eq:pos6}, \eqref{eq:pos7}, \eqref{eq:pos9} and \eqref{eq:pos11}, the closed loop error dynamics for error system \eqref{eq:poser5} can be given as 
\begin{align}\label{eq:poser6}
\begin{bmatrix}
    \dot{e}_{g_x}\\
    \dot{e}_{\ddot{x}_d} \\
    \dot{e}_{g_v} \\
    \dot{e}_{g_1}
\end{bmatrix} = \begin{bmatrix}
    -\frac{1}{\gamma_1}e_{g_x} + \ddot{x}_d \\
    -\frac{1}{\gamma_1}e_{\ddot{x}_d} + \dddot{x}_{d} \\
    -\frac{1}{\gamma_2}e_{g_v} +\frac{1}{\gamma_1}e_{\ddot{x}_d} \\
     -\frac{1}{\gamma_{21}}e_{g_1} + \frac{1}{\gamma_{21}}e_{g_v} + \frac{1}{\gamma_{21}}{e}_{\ddot{x}_d} + \Phi_{g_1} + \dddot{x}_{d},
\end{bmatrix}
\end{align}
where 
\begin{align} \label{eq:poser7}
    \Phi_{g_1} = 
\frac{\nabla{f(g_1)}(\nabla{f(g_1)})^T}{||\nabla{f(g_1)}||^2}\phi f(g_1) 
\end{align}
which appears only if $f(g_1)>0 \ \& \  \phi^T\nabla{f(g_1)}>0 $. Moreover, ${e}_{g_1}^T\Phi_{g_1} \leq 0$ (Please refer to the properties of projection in \cite{projection}).

\subsection{Attitude Tracking Error System}
The time derivatives of the attitude error system with its properties are discussed in detail in Appendix \ref{secA1}, which will be helpful in the theoretical claims provided in this section. For simplicity in the representation, we will use the following notations for the various attitude and angular velocity errors throughout the paper.
\begin{align} \label{eq:atter5}
{e}_R(R,R_d) = e_{R_d}, \quad {e}_R(R,R_c) = e_{R_c}, \quad {e}_R(R_d,R_c) = e_{R_{dc}}, \quad  e_{\Omega}(R,\Omega,R_d,\Omega_d) = e_{\Omega_d}.
\end{align}

Using the control input $M$, the closed loop error dynamics of angular error $e_{\Omega_d}$ premultiplied with inertia matrix $J$ can be given as 
\begin{align} \label{eq:atter6}
 J\dot{e}_{\Omega_d} = -k_Re_{R_d} - k_{\Omega}e_{\Omega_d}.
\end{align}

\subsection{Stability Analysis}
The stability and convergence analysis is systematically presented in this section. We first state the following lemmas, which will be utilized in the subsequent analysis.  

\begin{lemma}\label{lemma1}
There exist positive constants $\alpha_1,\alpha_2$ such that 
\begin{align} \label{eq:lemma1}
    \|f_d\| \leq \alpha_1, \quad \|f_d\| \geq \alpha_2.
\end{align}
Further, $\alpha_2$ is strictly positive if 
\begin{align}
k_{\alpha}<g-h_2-2 \label{Gain_K}.
\end{align}
\end{lemma}

\begin{proof} 
    Please refer to Appendix \ref{secA2} for the proof.
\end{proof}

\begin{lemma}\label{lemma2}
The quantity $\Delta{f} = f_d - fRe_3$ satisfies the following bound 
\begin{align} \label{eq:lemma3}
\|\Delta{f}\| \leq 
    2\alpha_1(\|e_{R_{dc}}\| + \|e_{R_d}\|).
\end{align}
\end{lemma}
\begin{proof}
    Please refer to Appendix \ref{secA3} for the proof. 
\end{proof}

\begin{lemma}\label{lemma3}
The desired angular velocity $\Omega_c$ satisfies the following bound
\begin{align} \label{eq:lemma5}
\|\hat{\Omega}_c\| \leq &  \rho_1(\|e_{\alpha}\|,\|\tanh(e_x)\|,\|\tanh(e_f)\|) + \rho_{01}.
\end{align}
\end{lemma}
\begin{proof}
    Please refer to Appendix \ref{secA4} for the proof. 
\end{proof}

\begin{lemma}\label{lemma8}
The filter errors given in equation \eqref{eq:poser6} can be bounded as 
\begin{align} 
    \|e_{g_x}(t)\| &\leq \sqrt{\|e_{g_x}(0)\|^2exp\left(-\frac{1}{\gamma_1}t\right) + \gamma_1^2h_2^2} \leq max(\|e_{g_x}(0)\|,\gamma_1h_2) = \alpha_3,  \label{eq:lemma81}\\
    \|e_{\ddot{x}_d}(t)\| & \leq  \sqrt{\|e_{\ddot{x}_d}(0)\|^2exp\left(-\frac{1}{\gamma_1}t\right) + \gamma_1^2h_3^2} \leq max(\|e_{\ddot{x}_d}(0)\|,\gamma_1h_3) = \alpha_4, \label{eq:lemma82}\\
    \|e_{g_v}(t)\|& \leq \sqrt{\|e_{g_v}(0)\|^2exp\left(-\frac{1}{\gamma_2}t\right) + \frac{\gamma_2^2}{\gamma_1^2}\|e_{\ddot{x}_d}\|^2} \leq max\bigg(\|e_{g_v}(0)\|,\frac{\gamma_2}{\gamma_1}max(e_{\ddot{x}_d}(0),\gamma_1h_3)\bigg) = \alpha_5, \label{eq:lemma83} \\
    \|e_{g_1}(t)\| & \leq max\left(\|e_{g_1}(0)\|,\sqrt{{8}(\alpha_4^2 +\alpha_5^2) + 2\gamma_{21}^2h_3^2}\right) = \alpha_6. \label{eq:lemma84}
\end{align}
\end{lemma}

\begin{proof}
    Please refer to Appendix \ref{secA5} for the proof.
\end{proof}

Next, we derive the stability and convergence properties of the attitude subsystem of the multirotor tracking the auxiliary desired attitude. 

\begin{theorem} \label{theorem1} (Adapted from \cite{Exponential_Geometric})
The closed loop error dynamics of $e_1(t) = [e_{R_d}(t)^T,e_{\Omega_d}(t)^T]^T$ is exponentially stable with the control input $M$ defined in \eqref{eq:att8} for the smooth desired attitude reference $\left(R_d(t), \ \Omega_d(t), \ \dot{\Omega}_d(t) \right)$. Moreover, an estimate of the region of attraction is given by 
\begin{align} 
    &\Psi(R(0),R_d(0)) < 2 , \label{eq:th11}\\
    &\|e_{\Omega_d}(0)\|^2 < \frac{2k_R}{\lambda_{max}(J)} (2 - \Psi(R(0),R_d(0)). \label{eq:th12}
\end{align}
\end{theorem}

\begin{proof}
    We start by proving that given condition \eqref{eq:th12}, the sublevel set $\mathcal{L}_2 = \{R,R_d \in \mathrm{SO(3)} | \Psi(R,R_d) < 2 \}$ is positively invariant set. Consider the following Lypunov 
    candidate 
    \begin{align} \label{eq:th13}
    V_1 = \frac{1}{2}e_{\Omega_d}^Te_{\Omega_d} + k_R\Psi(R,R_d)
    \end{align}
    
    Taking the time derivative of $V_1(t)$ after substituting the dynamics \eqref{eq:atter6} and \eqref{eq:A4} yields 
    \begin{align} \label{eq:th14}
    \dot{V}_1 = -k_{\Omega}\|e_{\Omega_d}\|^2.
    \end{align}

    Equation \eqref{eq:th14} along with \eqref{eq:th11} and \eqref{eq:th12} implies that 
    \begin{align} \label{eq:th15}
    k_R\Psi(R(t),R_d(t))\leq V_1(t) \leq V_1(0) < 2k_R \quad \forall t>0.
    \end{align}
    This proves that the sub-level set $\mathcal{L}_2$ is a positively invariant set such that $\Psi(R(t),R_d(t)) <2$ and hence the attitude errors are well defined. 
    
    To prove the exponential stability, let's consider another Lyapunov candidate
    \begin{align} \label{eq:th16}
    V_2 = \frac{1}{2}e_{\Omega_d}^Te_{\Omega_d} + k_R\Psi(R,R_d) + c_1e_{\Omega_d}^Te_{R_d},
    \end{align}
    where $c_1$ is a positive constant. Using the properties of configuration error function, for $e_1(t) = [e_{R_d}(t)^T,e_{\Omega_d}(t)^T]^T$ we get 
    \begin{align}\label{eq:th17}
        e_1^TW_{11}e_1 \leq V_1 \leq e_1^TW_{12}e_1,
    \end{align}
    where 
    \begin{align}\label{eq:th18}
        W_{11} = \begin{bmatrix}
        k_R & \frac{c_2}{2} \\
        \frac{c_2}{2} & \frac{\lambda_{min}(J)}{2}
    \end{bmatrix}, \quad     W_{12} = \begin{bmatrix}
        2k_R & \frac{c_2}{2} \\
        \frac{c_2}{2} & \frac{\lambda_{max}(J)}{2}
    \end{bmatrix}.
    \end{align}
    Taking time derivative of $V_2(t)$ and substituting the appropriate error dynamics using \eqref{eq:atter6}, \eqref{eq:A4} and \eqref{eq:A5} will yield
    \begin{align}\label{eq:th19}
        \dot{V}_2 = &- k_{\Omega}e_{\Omega_d}^Te_{\Omega_d} - c_1k_Re_{R_d}^TJ^{-1}e_{R_d} + c_1k_{\Omega}e_{R_d}^TJ^{-1}e_{\Omega_d} + c_1e_{\Omega_d}^TE^T(R,R_d)e_{\Omega_d}.
    \end{align}

    Further taking bounds using the properties in Appendix \ref{secA1} in equation \eqref{eq:th19} will yield 
    \begin{align}\label{eq:th20}
        \dot{V}_2 \leq &-\underbrace{\left[k_{\Omega}-\frac{c_1}{2}-\frac{c_1k_{\Omega}}{2\lambda_{min}(J)}\right]}_{\bar{\lambda}_1}\|e_{\Omega_d}\|^2 - \underbrace{\left[\frac{c_1k_R}{\lambda_{max}(J)} - \frac{c_1k_{\Omega}}{2\lambda_{min}(J)}\right]}_{\bar{\lambda}_2}\|e_{R_d}\|^2.
    \end{align}
    Furthermore, we can write 
    \begin{align}\label{eq:th21}
        \dot{V}_2 \leq & - \bar{\beta} V_2,
    \end{align}
    where $\bar{\beta} = \frac{min(\bar{\lambda}_1,\bar{\lambda}_2)}{\lambda_{max}(W_{12})}$. We choose the parameter so that $\bar{\beta}$ is positive. Then, the relation \eqref{eq:th21} invokes global uniform exponential stability (based on Theorem 4.10 in \cite{Khalil}).
\end{proof}

Now, we are ready to state our main results.

\begin{theorem} \label{theorem2}
Consider the multi-rotor dynamics \eqref{eq:dyn1}, \eqref{eq:dyn2}, \eqref{eq:dyn3} and \eqref{eq:dyn4}, with control inputs \eqref{eq:pos4} and \eqref{eq:att8}. The closed loop dynamics of overall system errors $e(t) = [e_1(t)^T,e_2(t)^T,e_3(t)^T]^T$ where $e_1(t) = [e_{R_d}(t)^T,e_{\Omega_d}(t)^T]^T$, $e_2(t) = [e_{R_{dc}}^T(t),e_{\Omega_{dc}}^T(t)]^T$ and $e_3(t) = [e_{\alpha}^T(t),\tanh(e_x(t))^T,\tanh(e_f(t))^T]^T$, is uniformly ultimately bounded, provided the following gain conditions are satisfied 
\begin{align}
mk_{\alpha}-\frac{m\alpha_{{x}}(3\alpha_{{x}}+1)}{2}-\frac{m\alpha_{f}}{2}- \rho_{03} >0&, \label{Gain1}  \\
m\alpha_x - \frac{m\alpha_{{x}}^2}{2} > 0&, \label{Gain2}  \\
\frac{1}{2\gamma_3}- {2{\alpha_1}^2} > 0&, \label{Gain3}  \\
\frac{m\alpha_{f}}{2}- \frac{m\alpha_{{x}}}{2}>0&, \label{Gain4}  \\
 k_{\Omega}-\frac{c_1}{2}-\frac{c_1k_{\Omega}}{2\lambda_{min}(J)} >0&, \label{Gain5}  \\
 \frac{c_1k_R}{\lambda_{max}(J)} - \frac{c_1k_{\Omega}}{2\lambda_{min}(J)} - {2{\alpha_1}^2} >0&, \label{Gain6} 
\end{align}
where the constant $\rho_{03}$ is defined subsequently. 
\end{theorem}

\begin{proof}
For the sake of simplicity, we proceed in the following steps to prove the theorem claims. 
\bmsubsubsection{Step 1}
Let's first consider the Lyapunov candidate $V_3(t)$ for the translational error $e_3(t)$ of multi-rotor 
\begin{align}\label{eq:stab_pos2}
    V_3 = &\frac{1}{2}me_{\alpha}^Te_{\alpha} + m\ln\left(\cosh(e_{x_1})\right) + m\ln \left (\cosh(e_{x_2})\right) + m\ln\left( \cosh(e_{x_3})\right) + \frac{1}{2}m\tanh(e_f)^T\tanh(e_f).
\end{align}
Taking the time derivative of $V_3(t)$ and substituting equations \eqref{eq:pos3}, \eqref{eq:poser1} and \eqref{eq:poser4} yields
\begin{align}\label{eq:stab_pos3}
\dot{V}_3 = & -mk_{\alpha}e_{\alpha}^Te_{\alpha} + e_{\alpha}^T\Delta{f} + e_{\alpha}^T\chi - me_{\alpha}^Te_{g_1} - me_{\alpha}^T\tanh(e_x) + mk_{\alpha}e_{\alpha}^T\tanh(e_f) \nonumber \\
& + m\tanh^T(e_x)\left(e_{\alpha} - \alpha_x\tanh(e_x) - \tanh(e_f)\right) + m\tanh(e_f)^T\left(-k_{\alpha}e_{\alpha} + \tanh(e_x) - \alpha_{f}\tanh(e_f)\right) 
\end{align}

Further solving \eqref{eq:stab_pos3} and taking bounds on $\dot{V}_2(t)$ yield 
\begin{align}\label{eq:stab_pos4}
\dot{V}_3 \leq & - mk_{\alpha}\|e_{\alpha}\|^2 - m\alpha_x\|\tanh(e_x)\|^2  - m\alpha_{f}\|\tanh(e_f)\|^2 + \|e_{\alpha}\|\|\Delta{f}\| + \|e_{\alpha}\|\|\chi\| + m\|e_{\alpha}\|\|e_{g_1}\|.
\end{align}


Solving equation \eqref{eq:stab_pos4} by substituting the dynamics of $\dot{e}_x(t)$ from \eqref{eq:poser1} in $\chi$, and using equations \eqref{eq:lemma3} and  \eqref{eq:lemma84} produces
\begin{align} \label{eq:stab_pos5}
\dot{V}_3 \leq & - mk_{\alpha}\|e_{\alpha}\|^2 - m\alpha_{f}\|\tanh(e_f)\|^2 + \|e_{\alpha}\|\left(m\alpha_{{x}}\|Sech^2(e_x)\| \|\tanh(e_f)\|
+ m\alpha_{f}\|\tanh(e_f))\|\right) - m\alpha_x\|\tanh(e_x)\|^2  \nonumber \\
& + \|e_{\alpha}\|\|Sech^2(e_x)\|\big(m\alpha_{{x}}\|e_{\alpha}\| + m\alpha_{{x}}^2 \|\tanh(e_x)\|\big) + m\|e_{\alpha}\|\alpha_6 + 2\alpha_1\|e_{\alpha}\|\big(\|e_{R_{dc}}\| + \|e_{R_d}\|\big),
\end{align}

which modifies to,
\begin{align}\label{eq:stab_pos6}
\dot{V}_3 \leq & - mk_{\alpha}\|e_{\alpha}\|^2 - m\alpha_x\|\tanh(e_x)\|^2  - m\alpha_{f}\|\tanh(e_f)\|^2 + m\alpha_{{x}}\|e_{\alpha}\|^2 + m\alpha_{{x}}^2 \|e_{\alpha}\|\|\tanh(e_x)\|  \nonumber \\
& + m\alpha_{{x}}\|e_{\alpha}\|\|\tanh(e_f)\| + m\alpha_{f}\|e_{\alpha}\|\|\tanh(e_f)\| + m\|e_{\alpha}\|\alpha_6 + 2\alpha_1\|e_{\alpha}\|\big(\|e_{R_{dc}}\| + \|e_{R_d}\|\big), \nonumber \\
\leq & - mk_{\alpha}\|e_{\alpha}\|^2 - m\alpha_x\|\tanh(e_x)\|^2  - m\alpha_{f}\|\tanh(e_f)\|^2 + m\alpha_{{x}}\|e_{\alpha}\|^2 + \frac{m\alpha_{{x}}^2}{2} \|e_{\alpha}\|^2 + \frac{m\alpha_{{x}}^2}{2}\|\tanh(e_x)\|^2 + \frac{m\alpha_{{x}}}{2}\|e_{\alpha}\|^2 \nonumber \\
&  + \frac{m\alpha_{{x}}}{2}\|\tanh(e_f)\|^2 + \frac{m\alpha_{f}}{2}\|e_{\alpha}\|^2 + \frac{m\alpha_{f}}{2}\|\tanh(e_f)\|^2 + \frac{m}{2}\|e_{\alpha}\|^2  + \frac{m}{2}\alpha_6^2 + \|e_{\alpha}\|^2 + {2{\alpha_1}^2}\|e_{R_{dc}}\|^2 + {2{\alpha_1}^2}\|e_{R_d}\|^2.
\end{align}

Further manipulation yield
\begin{align}\label{eq:stab_pos7}
\dot{V}_3 \leq & -\bigg(mk_{\alpha}-m\alpha_{{x}}-\frac{m\alpha_{{x}}^2}{2}-\frac{m\alpha_{{x}}}{2}-\frac{m\alpha_{f}}{2} -\frac{m}{2} - 1\bigg)\|e_{\alpha}\|^2 -\left(m\alpha_x - \frac{m\alpha_{{x}}^2}{2}\right)\|\tanh(e_x)\|^2 \nonumber \\
& - \left(m\alpha_{f} - \frac{m\alpha_{{x}}}{2} - \frac{m\alpha_{f}}{2}\right)\|\tanh(e_f)\|^2 + {2{\alpha_1}^2}\|e_{R_{dc}}\|^2 + {2{\alpha_1}^2}\|e_{R_d}\|^2 + \frac{m}{2}\alpha_6^2. 
\end{align}

\bmsubsubsection{Step 2} 
Now, let's consider another Lyapunov candidate $V_4(t)$ for auxiliary attitude errors $e_2(t)$ as 

\begin{align}\label{eq:Aux1}
V_4(t) = \Psi(R_d,R_c) + \frac{1}{2}e_{\Omega_{dc}}^Te_{\Omega_{dc}},
\end{align}

Taking the time derivative and using equations \eqref{eq:att3}, \eqref{eq:att4}, \eqref{eq:A4}, \eqref{eq:A5} and \eqref{eq:A6} yields 
\begin{align}\label{eq:Aux2}
\dot{V}_4(t) = & -\frac{1}{\gamma_3}e_{R_{dc}}^Te_{R_{dc}} - -\frac{1}{\gamma_4}e_{{\Omega}_{dc}}^Te_{{\Omega}_{dc}} - e_{R_{dc}}^TR_d^TR_c\Omega_c + \frac{1}{\gamma_3}e_{{\Omega}_{dc}}^TE(R_d,R_c)R_d^TR_c\Omega_c
\end{align}

Taking bounds on $\dot{V}_4(t)$ yields
\begin{align}\label{eq:Aux3}
\dot{V}_4(t) \leq &-\frac{1}{2\gamma_3}\|e_{R_{dc}}\|^2 -\frac{1}{2\gamma_4}\|e_{{\Omega}_{dc}}\|^2 + \left(\frac{\gamma_3 }{2}+ \frac{\gamma_4}{8\gamma_3^2}\right)\|\Omega_c\|^2
\end{align}

\bmsubsubsection{Step 3} 
To prove the stability of dynamics of error $e(t)$, let's consider the Lyapunov candidate as 
\begin{align}\label{eq:th211}
    V = & V_2(t) + V_3(t) + V_4(t)
\end{align}

Taking the derivative of $V(t)$ using \eqref{eq:Aux3}, \eqref{eq:stab_pos7} and \eqref{eq:th20} yields

\begin{align}\label{eq:th22}
\dot{V} \leq & -\left(mk_{\alpha}-m\alpha_{{x}}-\frac{m\alpha_{{x}}^2}{2}-\frac{m\alpha_{{x}}}{2}-\frac{m\alpha_{f}}{2} -\frac{m}{2} - 1\right)\|e_{\alpha}\|^2 -\left(m\alpha_x - \frac{m\alpha_{{x}}^2}{2}\right)\|\tanh(e_x)\|^2 \nonumber \\
& - \left(m\alpha_{f} - \frac{m\alpha_{{x}}}{2} - \frac{m\alpha_{f}}{2}\right)\|\tanh(e_f)\|^2 + \frac{m}{2}\alpha_6^2 -\left(\frac{1}{2\gamma_3}- {2{\alpha_1}^2}\right)\|e_{R_{dc}}\|^2 + \left(\frac{\gamma_3 }{2}+ \frac{\gamma_4}{8\gamma_3^2}\right)\|\Omega_c\|^2 \nonumber \\
& -\left(k_{\Omega}-\frac{c_1}{2}-\frac{c_1k_{\Omega}}{2\lambda_{min}(J)}\right)\|e_{\Omega_d}\|^2 -\frac{1}{2\gamma_4}\|e_{{\Omega}_{dc}}\|^2 - \left(\frac{c_1k_R}{\lambda_{max}(J)} - \frac{c_1k_{\Omega}}{2\lambda_{min}(J)} - {2{\alpha_1}^2}\right)\|e_{R_d}\|^2
\end{align}

From lemma \eqref{lemma3} and using the properties of hyperbolic functions, we can state that for some $\|e_{\alpha}\|< \bar{e}_\alpha$, the function $\rho_1()$ is locally Lipschitz. Thus, 
\begin{align}\label{eq:th23}
    \|\hat{\Omega}_c\|^2 \leq L_1^2\|e_\alpha\|^2 + \rho_{02}^2,
\end{align}
where $L_1$, a function of $\bar{e}_\alpha$ is the Lipschitz constant and $\rho_{02}$ is a positive constant. Substituting the \eqref{eq:th23} in \eqref{eq:th22} yields 
\begin{align}\label{eq:th24}
\dot{V} \leq & -\underbrace{\left[mk_{\alpha}-\frac{m\alpha_{{x}}(3\alpha_{{x}}+1)}{2}-\frac{m\alpha_{f}}{2}- \rho_{03}\right]}_{\lambda_1}\|e_{\alpha}\|^2 -\underbrace{\left[m\alpha_x - \frac{m\alpha_{{x}}^2}{2}\right]}_{\lambda_2}\|\tanh(e_x)\|^2 -\underbrace{\left[\frac{1}{2\gamma_3}- {2{\alpha_1}^2}\right]}_{\lambda_4}\|e_{R_{dc}}\|^2  - \underbrace{\left[\frac{1}{2\gamma_4}\right]}_{\lambda_5}\|e_{{\Omega}_{dc}}\|^2  \nonumber \\
&- \underbrace{\left[ \frac{m\alpha_{f}}{2}- \frac{m\alpha_{{x}}}{2}\right]}_{\lambda_3}\|\tanh(e_f)\|^2 -\underbrace{\left[k_{\Omega}-\frac{c_1}{2}-\frac{c_1k_{\Omega}}{2\lambda_{min}(J)}\right]}_{\lambda_6}\|e_{\Omega_d}\|^2 - \underbrace{\left[\frac{c_1k_R}{\lambda_{max}(J)} - \frac{c_1k_{\Omega}}{2\lambda_{min}(J)} - {2{\alpha_1}^2}\right]}_{\lambda_7}\|e_{R_d}\|^2 \nonumber \\
& + \underbrace{\left[\left(\frac{\gamma_3 }{2}+ \frac{\gamma_4}{8\gamma_3^2}\right)\rho_{02}^2 + \frac{m}{2}\alpha_6^2\right]}_{D},  
\end{align}
where $\rho_{03} = \frac{m}{2} + 1 + \frac{4\gamma_3^2 + \gamma_4}{8\gamma_3^2}L_1^2$. Equation \eqref{eq:th24} implies  
\begin{align}\label{eq:th25}
\dot{V} \leq & -\underbrace{min(\lambda_1,\lambda_2,\lambda_3,\lambda_4,\lambda_5,\lambda_6,\lambda_7)}_{\lambda_V}\|e\|^2 + D 
\end{align}

The relation \eqref{eq:th25} invokes uniformly ultimately bounded (UUB) stability since for $\|e(t)\|^2 > \frac{D}{\lambda_V}$, $\dot{V}(t)$ is negative-definite. However, note that \eqref{eq:th25} is derived based on \eqref{eq:th23}, which demands $\|e_\alpha\|\leq\bar{e}_\alpha$.  Hence, $\bar{e}_\alpha$ should be chosen such that $\bar{e}_{\alpha} \geq \sqrt{2\frac{max(V(0),\frac{D}{\lambda_V})}{m}}$, and accordingly the gain condition related to $\lambda_1$ has to be satisfied. This would ensure $\|e_\alpha(t)\|\leq\bar{e}_\alpha$, $\forall t\geq 0$, and hence, the UUB result would follow. 

\end{proof}

\begin{remark}
 It's worth noting that $D$ can be made arbitrarily small by choosing the filter gains $\gamma_1, \gamma_2, \gamma_{21}, \gamma_3, \gamma_4$ appropriately. This relaxes the unnecessary burden on the controller gains for such formulations. Moreover, if $k_{\alpha}$ is chosen such that the condition \eqref{Gain_K} is also satisfied then the desired attitude $R_c(t)$ always exists as $\|f_d(t)\| \geq \alpha_2 >0$ is guaranteed.
\end{remark}

\section{Conclusion}\label{Sec:conclusion}
In this paper, we introduced a trajectory-tracking controller for multi-rotor unmanned aerial vehicles (UAVs) with limited knowledge of the desired trajectory, restricted to position and heading information. Using the auxiliary filter dynamics, the unavailability of the higher derivative of the desired trajectory is tackled to design both attitude and translational motion controllers hierarchically. It is theoretically shown that the designed thrust is saturated and always positive, avoiding singularities in the desired attitude construction. The designed controller has proven to be ultimately bounded in the extended error space considering attitude and position error dynamics through rigorous theoretical analysis.

Future work would focus on experimental validation and real-world implementations to further validate the controller's robustness and adaptability. Additionally, exploring potential enhancements to accommodate a desired trajectory with unbounded derivatives (i.e. relaxing Assumption \ref{assump1}) and external disturbances could further broaden the scope of applications for multi-rotor UAVs.



\bmsection*{Acknowledgments}
This work is supported by TCS Research, India, under the TCS Research Scholar Program FP-140.



\bmsection*{Conflict of interest}
The authors declare no potential conflict of interests.

\bibliography{wileyNJD-AMA}



\appendix
\bmsection{Preliminaries on error formulation on $\mathrm{SO(3)}$}\label{secA1}
To design the attitude tracking controller on a nonlinear space $\mathrm{SO(3)}$, we start with constructing a positive definite configuration error function $\Psi$. We then choose the configuration error vector and a velocity error vector in the tangent space $T_R\mathrm{SO(3)}$ utilizing the derivatives of $\Psi$. Considering the carefully crafted error system, we design the controller following the Lyapunov direct method. Owing to the superior tracking performance for large initial attitude errors, we consider the following configuration error function \cite{Exponential_Geometric} $\Psi : \mathrm{SO(3)} \times \mathrm{SO(3)} \rightarrow \mathbb{R}$ for a given pair of desired attitude and angular velocity $R_1 \in \mathrm{SO(3)}, \Omega_1 \in \mathbb{R}^3 $, and the current attitude and angular velocity $R_2 \in \mathrm{SO(3)}, \Omega_2 \in \mathbb{R}^3$
\begin{align} \label{eq:A1}
    \Psi(R_2,R_1) = 2 - \sqrt{1 + tr[R_1^TR_2]}.
\end{align}

Based on the configuration error function, we define an attitude error vector $e_R(R_1,R_2): \mathrm{SO(3)} \times \mathrm{SO(3)} \rightarrow \mathbb{R}^3$ and an angular velocity error vector $e_{\Omega}: SO(3) \times \mathbb{R}^3 \times \mathrm{SO(3)} \times \mathbb{R}^3 \rightarrow \mathbb{R}^3$ as 

\begin{align}
& e_R(R_2,R_1) = \frac{1}{2\sqrt{1 + tr[R_1^TR_2]}}(R_1^TR_2 - R_2^TR_1)^\vee , \label{eq:A2}\\
& e_{\Omega}(R_2,\Omega_2,R_1,\Omega_1) = \Omega_2 - R_2^TR_1\Omega_1. \label{eq:A3}
\end{align}

Throughout this article, we will use the following well-established properties \cite{Exponential_Geometric} of these error variables in a sub-level set  $\mathcal{L}_2 = \{R_1,R_2 \in \mathrm{SO(3)} | \Psi(R_2,R_1) < 2 \}$ of $\Psi$
\begin{itemize} \label{item}
    \item $\Psi$ is positive definite about $R_2 = R_1$ and has only one critical point $R_2 = R_1$.
    \item $\|e_R\|^2 \leq \Psi \leq 2 \|e_R\|^2$.
    \item $\|e_R\| = sin\frac{\|\theta\|}{2}$, where $exp(\hat{\theta}) = R_1^TR_2$ i.e.  $\|e_R\| < 1$.
\end{itemize}

The time derivative  of the configuration function, attitude error and angular velocity  $\Psi(R_2,R_1),e_R(R_2,R_1),e_{\Omega}(R_2,\Omega_2,R_1,\Omega_1)$ can be given as
\begin{align}
&\frac{d}{dt}{\Psi}(R_2,R_1) = e_R(R_2,R_1)^Te_{\Omega}(R_2,\Omega_2,R_1,\Omega_1), \label{eq:A4} \\
&\frac{d}{dt}{e}_R(R_2,R_1) = E(R_2,R_1)e_{\Omega}(R_2,\Omega_2,R_1,\Omega_1), \label{eq:A5} \\
&\frac{d}{dt}J{e}_{\Omega}(R_2,\Omega_2,R_1,\Omega_1) = J\dot{\Omega}_2 + J(\hat{\Omega}_2R_2^TR_1\Omega_1 - R_2^TR_1\dot{\Omega}_1), \label{eq:A6} 
\end{align}

where $E(R_2,R_1) \in \mathbb{R}^{3 \times 3}$ is given as 
\begin{align}
 E(R_2,R_1) = &\frac{1}{2\sqrt{1 + tr[R_1^TR_2]}}\left(tr[R_2^TR_1]I - R_2^TR_1 + 2e_Re_R^T\right).  \label{eq:A7} 
\end{align}

Using the fact that $\|E(R_2,R_1)\| = \frac{1}{2}$ from  \cite{Exponential_Geometric}, we can state that 
\begin{align}
 \|\frac{d}{dt}{e}_R(R_2,R_1)\| \leq \frac{1}{2}\|e_{\Omega}(R_2,\Omega_2,R_1,\Omega_1)\|. \label{eq:atter4} 
\end{align}

\bmsection{Proof of Lemmas}

\bmsubsection{Lemma~{\rm\ref{lemma1}}}\label{secA2}
\begin{proof}
From equation \eqref{eq:pos5} we can write 
 \begin{align}\label{eq:lemma11}
     \|f_d\| \leq mg + m\|g_1\| + 2m\|\tanh(e_x)\| + mk_{\alpha} \|\tanh(e_f)\|.
 \end{align}

From the projection based differential estimate given in equation \eqref{eq:pos11}, we have $\|g_1\| \leq h_2$ and the fact that $\|\tanh(v)\| \leq 1 \ \forall v \in (- \infty, \infty)$, we can write

\begin{align}\label{eq:lemma12}
     \|f_d\| \leq mg + mh_2 + 2\sqrt{3}m + m\sqrt{3}k_{\alpha} = \alpha_1.
\end{align}

Also, from the third row of equation \eqref{eq:pos9} we get
\begin{align}\label{eq:lemma13}
    f_{d_z} = mg - m{g}_{1_z} + 2m\tanh(e_{z_1}) - mk_{\alpha}\tanh(e_{f_z})
\end{align}
Further calculating for the lower bound of the $f_{d_z}(t)$, we have   
\begin{align}\label{eq:lemma14}
    \|f_{d_z}\| \geq mg - mh_2 - 2m - mk_{\alpha}  = \alpha_2.
\end{align}
 $\alpha_2$ is strictly positive for the given choice of $k_{\alpha}$. 
\end{proof}

\bmsubsection{Lemma \ref{lemma2}}\label{secA3}
\begin{proof}
Substituting the expression of $f(t)$ in the expression of $\Delta{f}$, yields
\begin{align} \label{eq:B1}
    \Delta{f} = \|f_d\|R_ce_3 - \|f_d\|\sqrt{\frac{1 + e_3^TR_c^TRe_3}{2}}Re_3,
\end{align}

The upper bound of the $\Delta{f}$ can be calculated as  
\begin{align} \label{eq:B2}
    \|\Delta{f}\| \leq \bigg\| \|f_d\|R_ce_3 - \|f_d\|\sqrt{\frac{1 + e_3^TR_c^TRe_3}{2}}Re_3 \bigg\|.
\end{align}

Substituting the bounds of $\|f_d\|$ from lemma \ref{lemma1} to equation \eqref{eq:B2}
\begin{align} \label{eq:B3}
    \|\Delta{f}\| \leq \alpha_1\bigg\| R_ce_3 - \sqrt{\frac{1 + e_3^TR_c^TRe_3}{2}}Re_3 \bigg\|.
\end{align}

The quantity $e_3^TR_c^TRe_3$ represents the cosine of the angle between vectors $R_ce_3$ and $Re_3$. Let the angle be $\theta^{'}$, then from equation \eqref{eq:B3}, we have  

\begin{align} \label{eq:B4}
    \|\Delta{f}\| \leq \alpha_1\sqrt{1 - 2cos(\frac{\theta^{'}}{2})cos(\theta^{'}) + cos^2(\frac{\theta^{'}}{2})}.
\end{align}

Using Rodrigues' formula, let's consider the following
\begin{align}
    R_c^TR = exp(\hat{\theta}_1), \quad R_c^TR_d = exp(\hat{\theta}_2) \quad  R_d^TR = exp(\hat{\theta}_3), \label{eq:B5}
\end{align}

In the sublevel set $\mathcal{L}_2$ where $\theta^{'} \leq \|\theta_1\|$ the following holds 
\begin{align} \label{eq:B6}
   \sqrt{1 - 2cos(\frac{\theta^{'}}{2})cos(\theta^{'}) + cos^2(\frac{\theta^{'}}{2})} \leq 2 sin(\frac{\|\theta_1\|}{2}) = 2\|e_{R_c}\|.
\end{align}

Also, from \eqref{eq:B5} we have 
\begin{align}\label{eq:B7}
    R_c^TR = (R_c^TR_d)(R_d^TR), \quad exp(\hat{\theta}_1) = exp(\hat{\theta}_2)exp(\hat{\theta}_3), \quad exp(\hat{\theta}_1) = exp(\hat{\theta}_2 + \hat{\theta}_3). 
\end{align}

From \eqref{eq:B7}, the following relation holds
\begin{align}\label{eq:B8}
\|\theta_1\| \leq \|\theta_2\| + \|\theta_3\|.
\end{align}

Since the attitude error is equal to $sin\frac{\|\theta_i\|}{2}$ and in the sub-level set $\mathcal{L}_2$ i.e. $0\leq\|\theta_i\|<\pi, \forall i = 1,2,3$, its monotonically increasing. Hence,
\begin{align}\label{eq:B9}
\|e_{R_c}\| \leq \|e_{R_{dc}}\| + \|e_{R_d}\|.
\end{align}

From equation \eqref{eq:B9} and \eqref{eq:B6}, the bounds of $\|f_d\|$ can be given as 
\begin{align}\label{eq:B10}
\|\Delta{f}\| \leq 2\alpha_1(\|e_{R_{dc}}\| + \|e_{R_d}\|),
\end{align}
\end{proof}

\bmsubsection{Lemma \ref{lemma3}}\label{secA4}

\begin{proof}
From \eqref{eq:att1} we have 
\begin{align} \label{eq:C1}
    \|\hat{\Omega}_c\| \leq \|\dot{R}_c\|.
\end{align}
The time derivative of $R_c$ can be written as 
\begin{align} \label{eq:C2}
    \dot{R}_c = &[\dot{x}_{B_c},\dot{y}_{B_c},\dot{z}_{B_c}],
\end{align}
where 
\begin{align} \label{eq:C3}
\dot{z}_{B_c} = &\frac{\dot{f}_d}{\|f_d\|} - \frac{f_d(f_d^T\dot{f_d})}{\|f_d\|^3}, \nonumber \\
\dot{y}_{B_c} = &\frac{\dot{A}}{\|A\|} - \frac{(A^T\dot{A})}{\|A\|^3}A, \quad   A = x_{B_d} \times z_{B_c} \nonumber \\
\dot{x}_{B_c} = &\dot{y}_{B_c} \times {z}_{B_c} + {y}_{B_c} \times \dot{z}_{B_c}.
\end{align}

We individually calculate the upper bounds of column vectors of $\dot{R}_c$. The bounds on $\dot{z}_{B_c}$ can be given as 
\begin{align} \label{eq:C4}
\|\dot{z}_{B_c}\| \leq 2 \frac{\|\dot{f_d}\|}{\|f_d\|}
\end{align}

Similarly, the bounds on $\dot{y}_{B_c}$ can be given as 
\begin{align} \label{eq:C5}
\|\dot{y}_{B_c}\| \leq 2 \frac{\|\dot{A}\|}{\|A\|}
\end{align}

Since $x_{B_d} \nparallel z_{B_c}$, we have 
\begin{align} \label{eq:C6}
\|\dot{y}_{B_c}\| \leq \frac{2}{\delta_A}(\|\dot{x}_{B_d}\| + \|\dot{z}_{B_c}\|),
\end{align}
where ${\delta_A} \geq \|x_{B_d} \times z_{B_c}\|$. Using, equation \eqref{eq:C3} and properties of $\mathrm{SO(3)}$ we can obtain the upper bound of $\dot{x}_{B_c}$ as 
\begin{align} \label{eq:C7}
\|\dot{x}_{B_c}\| \leq  \|\dot{z}_{B_c}\| + \|\dot{y}_{B_c}\|
\end{align}

Now, the upper bound of $\|\hat{\Omega}_c\|$ can be written as 
\begin{align} \label{eq:C8}
\|\hat{\Omega}_c\| &\leq  \|\dot{z}_{B_c}\| + \|\dot{y}_{B_c}\| + \|\dot{x}_{B_c}\| \nonumber \\
&\leq  2\|\dot{z}_{B_c}\| + 2\|\dot{y}_{B_c}\| \nonumber \\
&\leq  (4 + \frac{4}{\delta_A})\frac{\|\dot{f_d}\|}{\|f_d\|} + \frac{4}{\delta_A}h_5. 
\end{align}

Let's evaluate the bound on the time derivatives of $f_d$. The time derivatives of $f_d$ can be given as 
\begin{align}\label{eq:C9}
    \dot{f}_d =  &-m\dot{g}_1 + 2mSech^2(e_x)\dot{e_x} - mk_{\alpha}Sech^2(e_f)\dot{e}_f.
\end{align}

Substituting the expressions from \eqref{eq:pos3}, \eqref{eq:pos11} and \eqref{eq:poser1} in the \eqref{eq:C9} will yield
\begin{align}\label{eq:C10}
     \dot{f}_d =  & -m\phi + m\Phi_{g_1} + 2mSech^2(e_x)\left(e_\alpha - \alpha_x\tanh(e_x) - \tanh(e_f)\right) + mk_{\alpha}^2e_{\alpha} - mk_{\alpha}\tanh(e_x) + mk_{\alpha} \alpha_{f}\tanh(e_f).
\end{align}

The quantity $\Phi_{g_1}$ can be bounded as  
\begin{align}\label{eq:C11}
 \|\Phi_{g_1}\| \leq &\frac{\|\nabla{f(g_1)}\|^2}{\|\nabla{f(g_1)}\|^2}\|\phi\| \|f(g_1)\|, \nonumber\\
 \leq & \|\phi\| \|f(g_1)\|,
\end{align}
and $\Phi_{g_1}$ appears when $f(g_1) >0 $ (i.e. $h_2 < \|g_1\| < h_2+ \epsilon_0$), then 
\begin{align}\label{eq:C12}
 \|\Phi_{g_1}\| \leq \|\phi\| \frac{\epsilon_0 + 2h_2}{h_2^2},
\end{align}
where $\phi$ can be written in the form of filter errors as 
\begin{align}\label{eq:C13}
 \phi = \frac{1}{\gamma_{21}}(e_{g_1} - e_{g_v} 
 - {e}_{\ddot{x}_d}).
\end{align}

Then from equation \eqref{eq:poser6}  and \eqref{eq:C10}, we can write 
\begin{align}\label{eq:C14}
 \|\dot{f}_d\| \leq & \frac{m}{\gamma_{21}}\left(1 + \frac{\epsilon_0 + 2h_2}{h_2^2}\right)(\|e_{g_1}\| + \|e_{g_v}\| 
 + \|{e}_{\ddot{x}_d}\|) + 2m\left(\|e_\alpha\| + \alpha_x \|\tanh(e_x)\| + \|\tanh(e_f)\|\right) \nonumber \\
     & + mk_{\alpha}^2\|e_{\alpha}\| + mk_{\alpha}\|\tanh(e_x)\| + mk_{\alpha}\alpha_{f}\|\tanh(e_f)\|.
\end{align}

Using lemma \ref{lemma1}, \ref{lemma8} and relation \eqref{eq:C14}, the equation \eqref{eq:C8} can be written as 
\begin{align} \label{eq:C15}
\|\hat{\Omega}_c\| \leq & \rho_1(\|e_{\alpha}\|,\|\tanh(e_x)\|,\|\tanh(e_f)\|) + \rho_{01},
\end{align}
where $\rho_1()$ is class $\kappa$-function and $\rho_{0_1}$ is a positive constant.
\end{proof}

\bmsubsection{Lemma \ref{lemma8}}\label{secA5}

\begin{proof}
For the error $e_{g_x}$, lets consider the following Lyapunov candidate
\begin{align} \label{eq:posf1}
    V_{g_x} = \frac{1}{2}e_{g_x}^Te_{g_x}.
\end{align}

Taking the time derivative of $V_{g_x}(t)$ and using \eqref{eq:poser7} will yield 
\begin{align} \label{eq:posf2}
    \dot{V}_{g_x}= -\frac{1}{\gamma_1}e_{g_x}^Te_{g_x} + e_{g_x}^T\ddot{x}_d.
\end{align}

This implies that 
\begin{align} \label{eq:posf3}
\|e_{g_x}\| = \sqrt{2{V}_{g_x}(t)} \leq \sqrt{\|e_{g_x}(0)\|^2exp\left(-\frac{1}{\gamma_1}t\right) + \gamma_1^2h_2^2}.  
\end{align}

Similarly, for error $e_{\ddot{x}_d}(t)$, consider the following Lyapunov candidate
\begin{align} \label{eq:posf4}
    V_{\ddot{x}_d} = \frac{1}{2}e_{\ddot{x}_d}^Te_{\ddot{x}_d}.
\end{align}

Taking the time derivative of $V_{\ddot{x}_d}(t)$ and using \eqref{eq:poser7} will yield 
\begin{align} \label{eq:posf5}
    \dot{V}_{\ddot{x}_d} = -\frac{1}{\gamma_1}e_{\ddot{x}_d}^Te_{\ddot{x}_d} + e_{\ddot{x}_d}^T\dddot{x}_d.
\end{align}

This implies that 
\begin{align} \label{eq:posf6}
\|e_{\ddot{x}_d}\| = \sqrt{2{V}_{\ddot{x}_d}(t)} \leq \sqrt{\|e_{\ddot{x}_d}(0)\|^2exp\left(-\frac{1}{\gamma_1}t\right) + \gamma_1^2h_3^2}.  
\end{align}

For error $e_{g_v}$, consider the following Lyapunov candidate
\begin{align} \label{eq:posf7}
    V_{g_v} = \frac{1}{2}e_{g_v}^Te_{g_v}.
\end{align}

Taking the time derivative of $V_{g_v}(t)$ and using \eqref{eq:poser7} will yield 
\begin{align} \label{eq:posf8}
    \dot{V}_{g_v} = -\frac{1}{\gamma_2}e_{g_v}^Te_{g_v} + \frac{1}{\gamma_1}e_{g_v}^Te_{\ddot{x}_d}.
\end{align}

This implies that 
\begin{align} \label{eq:posf9}
\|e_{g_v}\| = \sqrt{2{V}_{g_v}(t)} \leq \sqrt{\|e_{g_v}(0)\|^2exp\left(-\frac{1}{\gamma_2}t\right) + \frac{\gamma_2^2}{\gamma_1^2}\|e_{\ddot{x}_d}\|^2}.  
\end{align}

Finally, we will analyze the dynamical behavior of the error $e_{g_1}(t)$ by choosing the following Lyapunov candidate
\begin{align} \label{eq:posf10}
    V_{g_1} = \frac{1}{2}e_{g_1}^Te_{g_1}.
\end{align}

Taking the time derivative of $V_{g_1}(t)$ and using \eqref{eq:poser7} will yield 
\begin{align} \label{eq:posf11}
    \dot{V}_{g_1} = &-\frac{1}{\gamma_{21}}e_{g_1}^Te_{g_1} + \frac{1}{\gamma_{21}}e_{g_1}^Te_{g_v} + \frac{1}{\gamma_{21}}e_{g_1}^T{e}_{\ddot{x}_d} + e_{g_1}^T\Phi_{g_1} + e_{g_1}^T\dddot{x}_{d}.
\end{align}

Using the fact that $e_{g_1}^T\Phi_{g_1} \leq 0$, the bounds on $\dot{V}_{g_1}(t)$ can be given as 
\begin{align} \label{eq:posf12}
    \dot{V}_{g_1} \leq -\frac{1}{2\gamma_{21}}\|e_{g_1}\|^2 + \frac{2}{\gamma_{21}}\|e_{g_v}\|^2 + \frac{2}{\gamma_{21}}\|{e}_{\ddot{x}_d}\|^2 + \gamma_{21}h_3^2.
\end{align}

Solving \eqref{eq:posf12} for ${V}_{g_1}(t)$ using equations \eqref{eq:lemma82}, \eqref{eq:lemma83} 
\begin{align} \label{eq:posf13}
    {V}_{g_1}(t) \leq {V}_{g_1}(0) exp(-\frac{1}{\gamma_{21}}t) + {4}\alpha_4^2 + {4}\alpha_5^2 + \gamma_{21}^2h_3^2.
\end{align}

This implies that 
\begin{align} \label{eq:posf14}
    \|e_{g_1}(t)\| = \sqrt{2{V}_{g_1}(t)} \leq \sqrt{\|e_{g_1}(0)\|^2 exp(-\frac{1}{\gamma_{21}}t) + \alpha_7},
\end{align}
where $\alpha_7 = {8}(\alpha_4^2 +\alpha_5^2) + 2\gamma_{21}^2h_3^2$.
\end{proof}


\end{document}